\documentclass[letterpaper]{eptcs}  
\usepackage{times}
\usepackage{bbm}
\usepackage{graphicx}
\usepackage{amsthm}
\usepackage{amssymb}
\usepackage{array}
\usepackage{color}
\usepackage{subfigure}
\usepackage{mathpartir}
\usepackage{MnSymbol}

\title{On the reaction time of some synchronous systems}
\author{Ilias Garnier, Christophe Aussagu\`es, Vincent David
\institute{CEA, LIST, Embedded Real Time Systems Laboratory\\ Point Courrier 94, Gif-sur-Yvette, F-91191 France}
\email{Firstname.Lastname@cea.fr}
\and
Guy Vidal-Naquet
\institute{SUPELEC Systems Sciences (E3S) \\ Computer Science Department \\
91192 Gif-sur-Yvette Cedex, France}
\email{Guy.Vidal-Naquet@supelec.fr}
}
\def\titlerunning{On the reaction time of some synchronous systems}
\def\authorrunning{I. Garnier, C. Aussagu\`es, V. David \& G. Vidal-Naquet}

\newtheorem{definition}{Definition}
\newtheorem{theorem}{Theorem}
\newtheorem{lem}[theorem]{Lemma}

\begin{document}
\maketitle

\newcommand{\redt}[1]{\textcolor{red}{#1}}
\newcommand{\greent}[1]{\textcolor{green}{#1}}

\newcommand{\agents}{\mathcal{A}}
\newcommand{\aindex}{\mathcal{A}_{ind}}
\newcommand{\inputs}{\mathcal{I}}
\newcommand{\multiset}[1]{\mathcal{B}(#1)}
\newcommand{\bfin}[1]{\mathcal{B}_{fin}(#1)}
\newcommand{\bottom}{\perp}
\newcommand{\real}{\mathbb{R}}
\newcommand{\boolset}{\mathbb{B}}
\newcommand{\natset}{\mathbb{N}}
\newcommand{\fin}[1]{\natset^{<#1}}
\newcommand{\singletonset}{\mathbbm{1}}
\newcommand{\realplus}{\mathbb{R}^+}
\renewcommand{\time}{\mathbb{T}}
\newcommand{\prd}[1]{\prod_{#1}}
\newcommand{\coprd}[1]{\coprod_{#1}}
\renewcommand{\vec}[1]{\overrightarrow{#1}}
\newcommand{\veclt}{<:}
\newcommand{\vecgt}{:>}
\newcommand{\inl}{\textbf{inl}~}
\newcommand{\inr}{\textbf{inr}~}
\newcommand{\matchwith}{\textbf{case}}

\newcommand{\pletin}[2]{\textbf{let}\ #1\ =\ #2\ \textbf{in}}
\newcommand{\plet}[1]{\textbf{let}\ #1\ =}
\newcommand{\pin}{\textbf{in}\ }
\newcommand{\ifthen}[1]{\textbf{if}\ #1\ \textbf{then}}
\newcommand{\pfin}[1]{\wp_{fin}(#1)}

\newcommand{\inttype}{\textbf{int}}
\newcommand{\booltype}{\textbf{bool}}
\newcommand{\commtype}{\textbf{comm}}
\newcommand{\vartype}[1]{\textbf{var($#1$)}}
\newcommand{\exptype}[1]{\textbf{exp($#1$)}}
\newcommand{\intvartype}{\textbf{var(int)}}
\newcommand{\boolvartype}{\textbf{var(bool)}}
\newcommand{\intexptype}{\textbf{exp(int)}}
\newcommand{\boolexptype}{\textbf{exp(bool)}}

\newcommand{\Askip}{\textbf{skip}}
\newcommand{\Atrue}{\textbf{tt}~}
\newcommand{\Afalse}{\textbf{ff}~}
\newcommand{\Awhile}[2]{\textbf{while}~#1~\textbf{do}~#2~\textbf{done}}
\newcommand{\Acond}[3]{\textbf{if}~#1~\textbf{then}~#2~\textbf{else}~#3}
\newcommand{\Anew}[3]{\textbf{new}~#1~:#2~\textbf{in}~#3}
\newcommand{\Aadvance}[2]{\textbf{advance}(#1,~#2)}
\newcommand{\Atick}[1]{\textbf{tick}(#1)}
\newcommand{\Aget}[1]{\textbf{get$_#1$}}
\newcommand{\Agetn}{\textbf{get}}

\newcommand{\Asend}[2]{\textbf{send$_#1$}(#2)}
\newcommand{\Atake}[1]{\textbf{take$_#1$}}

\newcommand{\Abs}[1]{\lambda #1 \cdot}

\newcommand{\machines}{\textsc{TM-Machine}}

\newcommand{\simulates}[2]{#2 \precsim #1}

\newcommand{\simin}[1]{\underset{#1}{\sim}}
\renewcommand{\simin}[2]{{({\scriptstyle{#1} \sim \scriptstyle{#2}})}}
\newcommand{\notsimin}[1]{\underset{#1}{\not \sim}}
\renewcommand{\notsimin}[2]{{({\scriptstyle{#1} \not \sim \scriptstyle{#2}})}}

\newcommand{\notsim}{\not \sim}

\newcommand{\boundedseparable}[1]{\overset{#1}{\nsim}}

\newcommand{\outmorph}{\mathsf{out}}
\newcommand{\nextmorph}{\mathsf{next}}

\newcommand{\edge}[1]{\overset{#1}{\rightarrow}}
\newcommand{\longedge}[1]{\xrightarrow{#1}}
\newcommand{\transedge}[1]{\xrightarrow{#1}^+}

\newcommand{\process}[2]{\textsc{Proc}(#1,#2)}
\newcommand{\processnopar}{\textsc{Proc}}

\newcommand{\RT}{\textsc{RT-Proc}}

\newcommand{\define}{\triangleq}

\newcommand{\rtmachine}{RT-Machine}

\newcommand{\reducesto}{\Downarrow}
\newcommand{\diverges}{\underset{\infty}{\Downarrow}}
\newcommand{\values}{\mathcal{V}}

\newcommand{\preonestep}{\rightarrowtail}
\newcommand{\step}[1]{\underset{\rightarrowtail}{\textsc{#1}}}
\newcommand{\onestep}{\rightarrow}
\newcommand{\finitestep}{\overset{*}{\rightarrow}}
\newcommand{\infinitestep}{\overset{\infty}{\rightarrow}}
\newcommand{\coinductive}{\overset{co*}{\rightarrow}}

\renewcommand{\sc}[1]{\textsc{#1}}

\newcommand{\pumpfailure}{{PUMP\_FAILURE}~}
\newcommand{\emstop}{{EM\_STOP}~}
\newcommand{\sbwaiting}{{SB\_WAITING}~}
\newcommand{\progready}{{PROGRAM\_READY}~}
\newcommand{\puready}{{PU\_READY}~}
\newcommand{\fmstartup}{{FM\_STARTUP}~}

\newcommand{\normstartup}{{NORM\_STARTUP}~}
\newcommand{\degrstartup}{{DEGR\_STARTUP}~}
\newcommand{\degrstartupreq}{{DEGR\_STARTUP\_REQ}~}

\newcommand{\constraint}[3]{#1 \overset{#3}{\leadsto} #2}
\newcommand{\machinerun}[2]{\xrightarrow[#2]{#1}}

\renewcommand{\bar}[1]{\overline{#1}}
\newcommand{\nondetseparators}{\mathsf{S}}
\newcommand{\detseparators}{\mathsf{DS}}
\newcommand{\prl}{\parallel}
\newcommand{\systems}[2]{\mathsf{Sys}(#1, #2)}

\newcommand{\datatypes}{\mathfrak{D}}
\newcommand{\abstractdatatypes}{\mathfrak{A}}

\newcommand{\prefix}{\mathsf{prefix}}

\renewcommand{\language}[1]{\mathcal{L}^\ast_{#1}}
\newcommand{\inflanguage}[1]{\mathcal{L}^\omega_{#1}}

\newcommand{\alllanguage}[1]{\mathcal{L}^{\infty}_{#1}}
\newcommand{\finlanguage}[1]{\mathcal{L}^{\ast}_{#1}}

\newcommand{\outputs}{\mathcal{O}}

\newcommand{\reactime}{\mathsf{reactime}}
\newcommand{\detreactime}{\mathsf{detreactime}}
\newcommand{\reactive}{\mathsf{reactive}}

\newcommand{\differences}{\mathsf{diff}}

\newcommand{\separatingpairs}{\mathsf{SepPairs}}
\newcommand{\detseppairs}{\mathsf{DSepPairs}}
\newcommand{\strongseppairs}{\mathsf{SSP}}

\newcommand{\detobs}{\mathsf{DOE}}

\newcommand{\sspsequence}{\mathsf{SSPseq}}

\newcommand{\merge}{\oplus}

\newcommand{\DOEorder}{\preccurlyeq}

\newcommand{\obsorder}{\mathsf{ObsOrder}}

\newcommand{\DOEcompose}{\mathsf{DOEcompose}}

\newcommand{\mapfunc}{\mathsf{map}}


\begin{abstract}
  This paper presents an investigation of the notion of reaction time in some synchronous systems.
  A  state-based description of such systems is given, and the reaction time of such systems
  under some classic composition primitives is studied.
  Reaction time is shown to be non-compositional in general. Possible solutions are
  proposed, and applications to verification are discussed. This framework is illustrated 
  by some examples issued from studies on real-time embedded systems.
\end{abstract}

\section{Introduction}

A primary concern when developing hard real-time embedded systems
is to ensure the timeliness of computations. This kind of requirement
is often expressed as a reaction time constraint, i.e. an upper bound on
the time the system may take to process an input and produce the related output.
When systems are composed of multiple communicating agents, this task may
be difficult. In this paper, we propose a formalization of reaction time for a certain class of
synchronous systems. We show that reaction time is a fine-grained notion of functional
dependency, and we show that it is non-compositional. In order to solve this
problem, we propose an approximate but compositional method to reason on
functional dependency and reaction time.

\paragraph{Related work.} The specification and verification of temporal properties
traditionally relies on temporal logic \cite{Emerson:1982} or related formalisms \cite{Alur94atheory}.
With these formalisms, the system designer gives a specification of some
causality or quantitative property which is then verified by model-checking.

These methods are also applicable to the restricted class of synchronous systems \cite{lustre:ieee}.
The OASIS \cite{DavidDLOHP98} system, which motivated this study, belongs to this class. A traditional compositional 
verification of synchronous systems using Moore machines \cite{MooreMachines} was given in \cite{Clarke1989}. 
Our formal framework to reason on reaction time was heavily inspired by the literature 
on information flow analysis \cite{Barbuti02secureinformation} and on the category-theoretic 
view of process algebras \cite{AbramskyInteraction}. 
It is also similar to \emph{testing} methods \cite{SimTestingStannett}.

\section{Preliminaries}

\subsection{Case study: proving the reactivity of a simple system}
\label{casestudy}

Let $S$ be a black box with two buttons $A$ and $B$ as inputs and
the elements of any non-singleton set as outputs. In this example, we will assume 
that $S$ is deterministic. Our goal is to decide whether pressing $A$ has
any observable effect on the system. A naive solution is to verify
whether the new observable state is different from the previous one.
$$
\begin{array}{l|lll}
  \text{Input}            & S & \longedge{A} & S_A \\
  \text{Observable state} & o &              & o'
\end{array}
$$
If $o \not = o'$, we may consider that the system seems to have answered
to the pressing of $A$. There are two counter-arguments to this conclusion.
\begin{enumerate}
\item $S$ may have decided in advance to output $o'$;
\item there is no reason for an observable consequence to occur immediately
  after pressing the button.
\end{enumerate}
In order to obtain a correct solution, the main point to take into account
is that the observable state is not only function of the inputs but also
of the internal state. From now on, we will assume that we have two
identical copies of $S$, and we will proceed to the experiment simultaneously
with the button $A$ and the button $B$.
$$
\begin{array}{l|lll|lll}
  \text{Input}            & S & \longedge{A} & S_A & S & \longedge{B} & S_B \\
  \text{Observable state} & o &              & o' & o &              & o''
\end{array}
$$
If $o' \not = o''$, we deduce that the system has distinguished between
pressing the button $A$ and the button $B$. This experiment is thus strictly more
informative than the previous one. In the other case, knowing that $o' = o''$
is not enough for us to extract any information on the internal behavior of the
system. Indeed, the second counter-argument advances that the observable reaction
can occur after an arbitrary number of transitions. A solution is to iterate the
experiment on $S_A$ and $S_B$ until observing a difference, but the observable
state then becomes possibly correlated to the other choices $A$ or $B$ performed
during the experiment. The choices must thus be identical for $S_A$ and $S_B$.

We can informally define reactivity by stating that if there exists a finite
sequence of experiments (i.e. a word on $\{A,B\}$) allowing to distinguish the
systems $S_A$ and $S_B$, then $S$ is reactive. In formal terms, this is equivalent to
stating that $S_A$ and $S_B$ must be non-bisimilar. The reaction \emph{time}
is the maximum length of the minimal experiment allowing to prove 
non-bisimilarity.

\subsection{Notations and definitions}

Some definitions will be useful to our work. Let $\Sigma$ be a set. The set of finite words on $\Sigma$ is noted 
$\Sigma^*$, and the set of infinite words is noted $\Sigma^\omega \equiv \natset \rightarrow \omega$.
The set of finite and infinite output words is $\Sigma^\infty = \Sigma^* \bigcup \Sigma^\omega$.
The length of a finite word $w$ will be noted $|w|$.
For any word $w$ we will note $\prefix(w, len)$ the prefix of $w$ of length $len$, and 
we also note $w[i]$ the $i$-th symbol of $w$, where $i \in [0; |w|-1]$.

The singleton set is $\singletonset = \{ \star \}$ (up to isomorphism), and the disjoint union of two
sets $A$ and $B$ is noted $A + B$. The set of natural integers strictly inferior to $x$ is noted $\fin{x}$.

\section{A formal model of synchronous systems}
 
We aim at giving a formal model of synchronous systems sufficiently expressive to 
encode languages such as Lustre \cite{lustre:ieee} and PsyC \cite{DavidDLOHP98}.
Our formalism is an adaptation of Moore machines \cite{MooreMachines}.

\subsection{The synchronous abstraction}

We will restrict ourselves to the set of systems which respect the synchronous mode of
computation. In this model, the computation is divided in successive \emph{rounds}.
Each transition from a round to the next denotes the \emph{tick} of a global logical clock.
The observable state of a system is constant on each round, and changes only at the
boundary between rounds. At each new round, the new internal state is a function of
the current input and internal state (equivalently, the internal
state is a function of the initial internal state and all previous inputs).
The observable state is only function of the internal state.
The following timeline shows an example of a deterministic synchronous computation 
involving three successive rounds.
\begin{center}
\includegraphics{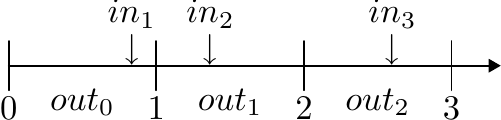}
\end{center}
In this example, the output $out_0$ is a function of the initial internal state only;
the output $out_1$ is a function of the initial internal state and $in_1$; and
the output $out_2$ is a function of the initial internal state, $in_1$ and $in_2$.
Thus, inputs have no \emph{immediate} effect on the observable state.

\subsection{State-based description of synchronous systems}

We define our synchronous systems as a labeled transition system (LTS)
inspired by Moore machines. We recall here some definitions which will be useful in the
following developments.

\begin{definition}[Synchronous system] Let $In$ be a set of inputs and $Out$
  be a set of inputs. A synchronous system $S = \langle In, Out, Q, E, \outmorph, q_i \rangle$ is the data of:
  \begin{itemize}
  \item a set of states $Q$,
  \item a transition relation $E \subseteq Q \times In \times Q$,
  \item a labeling function associating states to outputs $\outmorph : Q \rightarrow Out$,
  \item and an initial state $q_i \in Q$. 
  \end{itemize}
  The sets $In$ and $Out$ are the signature of $S$. We will note $p \longedge{a} q$ as a shorthand
  for $(p, a, q) \in E$. Moreover, we constrain our systems to be \emph{finitely branching}
  and to be \emph{complete}, i.e. $\forall p \in Q, \forall a \in In, \exists p \longedge{a} q$
  \footnote{In practice, these conditions constrain the input and output data sets to be finite.}.
\end{definition}

The computational meaning of a LTS is expressed using the notion of \emph{run}. It allows 
to define the \emph{output language} associated to an input word.

\begin{definition}[Run of a synchronous system, output language]
Let $S = \langle In, Out, Q, E, \outmorph, q_i \rangle$ be a synchronous system.
We define the notions of \emph{finite} run and the associated output language.

\textbf{Finite runs.} Let $w \in In^*$ be a finite input word. The set of finite, maximal runs of $S$ on $w$
starting from state $q_0 \in Q$ is $Runs^*_S(q_0, w) \subseteq Q \times (In \times Q)^*$ and is defined as:
$$
Runs^*_S(q_0, w) = \{ q_0 . w[0] . q_1 \ldots w[|w|-1] . q_{|w|-1}~|~ \forall i \in [0; |w|-1], q_i \longedge{w[i]} q_{i+1} \}.
$$

\textbf{Output language.} The output language of $S$ associated to $w$ and $q_0$ is 
$\finlanguage{S}(q_0, w) \subseteq Out^\ast$ and is defined as follows:
$$
\finlanguage{S}(q_0, w) = \{ \outmorph(q_0) . \outmorph(q_1) \ldots \outmorph(q_{n-1})~|~ q_0 . a_0 . q_1 . a_1 \ldots \in Runs^\ast_S(q_0,w) \}.
$$
\end{definition}

The classical equivalence relation on states of labeled transition systems is bisimilarity.
Its definition is slightly adapted to our notion of synchronous system.

\begin{definition}[Bisimilarity, non-bisimilarity] Let $S = \langle In, Out, Q, E, \outmorph, q_i \rangle$ be a
  synchronous system. A relation $R \subseteq Q \times Q$ is said to be a (strong) \textbf{bisimulation}
  if and only if the following condition holds:
  $$
    \forall (p, q) \in R,  \outmorph(p) = \outmorph(q)~\wedge~ 
                           (\forall p \longedge{a} p', \exists q \longedge{a} q', (p', q') \in R)~ \wedge~
                           (\forall q \longedge{a} q', \exists p \longedge{a} p', (p', q') \in R).
  $$
  If there exists such a relation $R$ s.t. $(p, q) \in R$, then $p$ and $q$ are said to be bisimilar,
  which is noted $p \sim q$. Moreover, bisimilarity is an equivalence relation.
  Conversely, the negation of bisimilarity $\nsim \subseteq Q \times Q$ is inductively defined by the rules below.
  In these rules, $p, q \in Q$ and $a \in In$ are universally quantified.
  $$
  \begin{array}{lr}
    \inferrule [base]
    { }
    { \outmorph(p) \neq \outmorph(q) \rightarrow p \nsim q } &
    \inferrule [ind]
    { \exists p \longedge{a} p', \forall q \longedge{a} q', p' \nsim q'~\vee~\exists q \longedge{a} q', \forall p \longedge{a} p', p' \nsim q' }
    { p \nsim q }
  \end{array}
  $$
\end{definition}

In this paper, except when stated otherwise, all state spaces shall be assumed to be quotiented by
bisimulation equivalence.

\section{Reaction time of a state in a synchronous system}

This section formalizes the ideas exposed in the case study, in Sec.\ref{casestudy},
and extends them to non-deterministic systems. The case study proposes to
model reaction as a functional dependency between a set of inputs and the
future behavior of the system. In our formal model, these future behaviors 
are represented as the successor states of the considered state.

In this setting, we will first define a notion of reactivity inspired by functional
dependency, and then define reaction time as the necessary time to prove
that two successor states are not bisimilar.

\subsection{Reactivity}
Let $In = \{A, B\}$ and $Out$ be a non-singleton set. Let $S = \langle In, Out, Q, \outmorph, E, q_i \rangle$ be a synchronous system.
and let's assume that state $q \in Q$ is as depicted in Fig. \ref{fig:somesystems:reactive}. We observe that
there are two inputs $A$ and $B$ leading to non-bisimilar states $q_1$ and $q_2$.
In this case, state $q$ is thus reactive. In the case of non-deterministic systems, 
we must generalize this idea: if there exists an asymmetry in the possible transitions of a system,
then it is reactive. Let's assume that state $q$ is as depicted in Fig. \ref{fig:somesystems:nonreactive}.
There, $q$ is not reactive because there is a symmetry between the transitions possible with
$A$ and the transitions possible with $B$. This symmetry is broken in Fig. \ref{fig:somesystems:nonreactive}.
Non-determinism highlights the fact that reactivity is a kind of non-bisimilarity.

\begin{figure}
  \centering
  \subfigure[Reactive deterministic system]{
    \includegraphics[height=3cm]{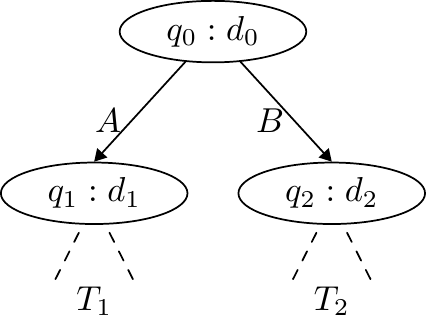}
    \label{fig:somesystems:reactive}
  }
  \subfigure[Non-reactive non-det. sys.]{
    \includegraphics[height=3cm]{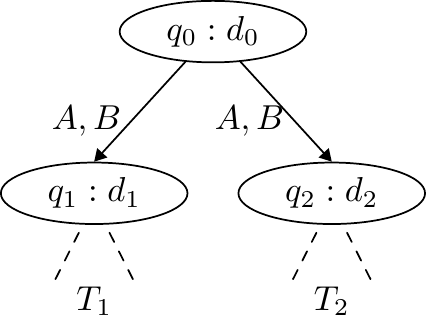}
    \label{fig:somesystems:nonreactive}
  }
  \subfigure[Reactive non-deterministic system]{
    \includegraphics[height=3cm]{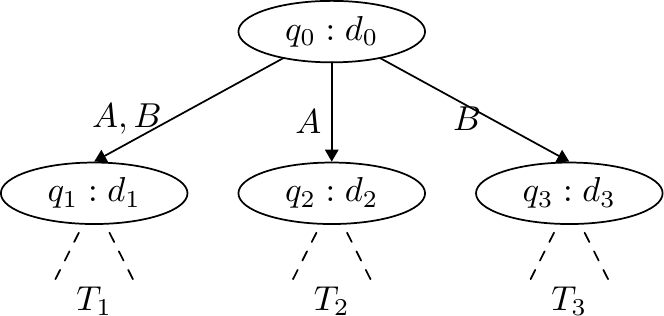}
    \label{fig:somesystems:reactive2}
  }
  \caption{Some reactive and non-reactive systems}
  \label{fig:somesystems}
\end{figure}

\begin{definition}[Reactivity of a state in a synchronous system, separating pair]
  Let $In, Out$ be two sets, $S = \langle In,$ $Out,$ $Q,$ $\outmorph,$ $E,$ $q_i \rangle$ be a synchronous system
  and $q \in Q$ be a state. We will denote $\reactive(q)$ the fact that $q$ is reactive. The
  predicate $\reactive(q)$ is defined as follows:
  $$
    \reactive(q) \define~\exists a_1, a_2 \in In, a_1 \not = a_2~\wedge~
    \left( \exists q \edge{a_1} q_1, \forall q \edge{a_2} q_2, q_1 \nsim q_2 \right).
  $$
  The pair of inputs $(a_1, a_2)$ is a \textbf{separating pair} of $q$. It is \emph{deterministic} iff
  $\forall q \edge{a_1} q_1, \forall q \edge{a_2} q_2, q_1 \nsim q_2$. The set of separating pairs of $q$ is noted 
  $\separatingpairs(q)$, and the deterministic subset is $\detseppairs(q)$.
  
\end{definition}

\subsection{Observable effects}

Observable effects stem from a fine-grained study of reactivity. In this section, we show that an observable effect
characterizes a temporally localized difference between the behaviors of non-bisimilar states.
We show that an input data has observable effects on the system on a not necessarily finite interval.

\paragraph{Characterizing the difference between two states.} 
Characterizing difference between states can be done by studying the negation of bisimulation.

\begin{definition}[Separators, strongly separable states]
  Let $S = \langle In, Out, Q, \outmorph, E, q_i \rangle$ be a synchronous system and $p_1, q_1 \in Q$ s.t. 
  $p_1 \nsim q_1$. A constructive proof of $p_1 \nsim q_1$ is the data of (at least) two 
  \emph{separating runs} $r_1 \in Runs^*_S(p_1, w)$ and $r_2 \in Runs^*_S(q_1, w)$:
  $$
    r_1 = p_1 . a_1 . p_2 . a_2 \ldots a_n . p_n~\text{and}~r_2 = q_1 . a_2 . q_2 . a_2 \ldots a_n . q_n.
  $$
  These runs are labelled on input by a finite word $w = a_1 . a_2 \ldots a_n$ called \textbf{separator}, and
  generate output words $o_1 \in Out^*$ $= \outmorph(p_1) .$ $\outmorph(p_2) \ldots$ $\outmorph(p_n)$
  and $o_2 \in Out^*$ $= \outmorph(q_1) .$ $\outmorph(q_2) \ldots$ $\outmorph(q_n)$.
  More generally, any separator $w$ induces a nonempty set $\outputs(p_1, q_1, w)$ of pairs of 
  different output words $(o_1, o_2)$ generated by separating runs s.t.
  $o_1 \in \finlanguage{S}(p_1, w), o_2 \in \finlanguage{S}(q_1, w)$.

  A separator $w$ is \textbf{deterministic} when all its runs are separating, i.e. all runs stems from
  a proof of $p_1 \nsim q_1$. The set of separators of two states $p, q$ is noted $\nondetseparators(p, q)$, 
  and the set of deterministic separators is noted $\detseparators(p, q)$. 
  Note that $\detseparators(p, q) \subseteq \nondetseparators(p, q)$.

  Two states $p, q$ are said to be \textbf{strongly separable}, noted $p \separated q$, iff all infinite inputs
  words are prefixed with a deterministic separator.
\end{definition}
\vspace{4mm}

\begin{figure}
  \centering
  \subfigure[]{
    \includegraphics[width=6.5cm]{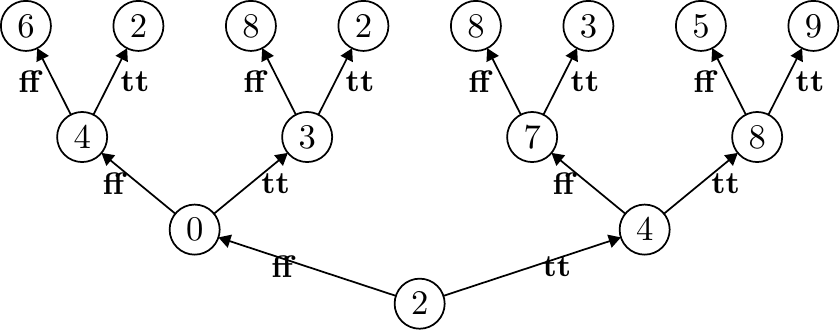}
    \label{fig:separators:left}
  }
  \subfigure[]{
    \includegraphics[width=6.5cm]{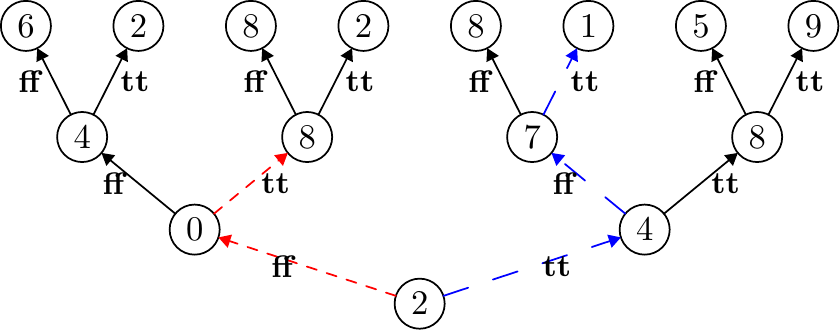}
    \label{fig:separators:right}
  }
  \caption{A pair of separable LTS}
  \label{fig:separators}
\end{figure}

Fig.~\ref{fig:separators} shows two separable LTS.
The fact that they are non-bisimilar is proved by the existence of two separators 
(although one would suffice) of length two and three, as emphasized by the dotted paths. 

Once separability of two states defined, we can define what is an observable effect
and \emph{when} it occurs. This is based on observing the differences in the output 
word pairs generated by a separator. 

\begin{definition}[Observable effect]
  Let $p, q  \in Q$ be two states. Let $w \in In^*$ be an input word.
  The observable effects are generated by all the prefixes of $w$ wich are \emph{separators}.
  $$
  \begin{array}{l}
    \differences_{p,q} : \prod_{w \in In^*} \fin{|w|} \rightarrow \singletonset + (Out \times Out) \\
    \differences_{p,q}(w, n) = (x_1, x_2) \leftrightarrow \prefix(w, n+1) \in \nondetseparators(p,q)~\wedge~  (o_1, o_2) \in \outputs(p, q, \prefix(w, n+1))~\wedge\\
    \hspace{4cm} x_1 = o_1[n] \wedge x_2 = o_2[n]
  \end{array}
  $$
  $\differences_{p, q}(w, n)$ returns $\star \in \singletonset$ if there is no difference at index $n$
  or a pair $(o_1[n], o_2[n])$ s.t. $o_1[n] \not = o_2[n]$ at the same index. These differences are the
  \textbf{observable effects} induced by $w$.
\end{definition}

\begin{figure}
  \centering
  \subfigure[]{
    \includegraphics[width=5.5cm]{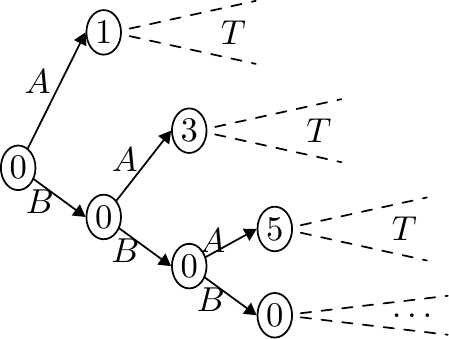}
    \label{fig:WCRT:left}
    \hspace{2cm}
  }
  \subfigure[]{
    \includegraphics[width=5.5cm]{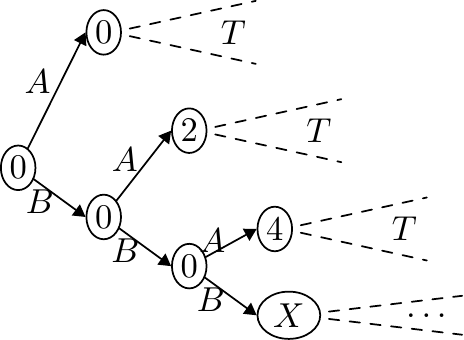}
    \label{fig:WCRT:right}
  }
  \caption{Partial separability}
  \label{fig:WCRT}
\end{figure}

The existence of a separator ensures that there \emph{may} be an observable effect.
Fig. \ref{fig:WCRT} show two systems (whose state space is quotiented by bisimulation equivalence)
with an unknown output data $X$ in Fig. \ref{fig:WCRT:right}.
State names are omitted and the nodes only contain the output data.
If $X = 0$, the initial states are separable with the words of the language $B^\ast A$. 
These are \emph{deterministic} separators: reading them on input ensures an observable effect.
We observe that the input words of the language $B^\infty$ do not yield an observable effect.
On the other hand, if $X \neq 0$, all infinite words are prefixed by a (deterministic) separator.
An eventual observable reaction is guaranteed.

\subsection{Reaction time}
\label{sec:reactiontime}

Since we work with logical time, the occurrence times of observable effects are their indices in the associated output traces.
We may \emph{define} the reaction time of a state as the maximum of the occurrence time of the
first observable effect. This yields two possible views of reaction time: an optimistic one 
(an observable effect \emph{may} arise \ldots) and a pessimistic one (an observable effect \emph{must} 
arise). Moreover, the reaction time of a state can be valid for \emph{all} contexts or just for \emph{some}.
Our application domain requires that we choose a pessimistic approach. Compositionality in turn
requires that we quantify over all possible contexts when defining reaction time, as will be shown later.

\begin{definition}[Deterministic reaction time]
The (deterministic) reaction time of a state w.r.t. an input is the \emph{maximum} number of transitions
that must be performed to see the first observable effect arise, for any input.
Let $q \in Q$ be a state s.t. $\reactive(q)$ holds. We note by $\detreactime(q) = t$ the fact that $q$ has 
a reaction time of $t$ transitions, where:
$$
\begin{array}{ll}
  \detreactime(q) = max \{ ~n ~~| & (a_1, a_2) \in \detseppairs(q),~q \edge{a_1} q_1,~q \edge{a_2} q_2, ~~ q_1 \separated q_2, \\
                                  & w \in In^\omega,~~ \differences_{q_1,q_2}(w, n) \neq \star~\wedge~\forall n' < n, \differences(w, n') = \star \} 
\end{array}
$$
\end{definition}

\section{Observable effects under composition}

In the previous section, we have defined a notion of observable effects for synchronous systems.
In this section, we will investigate the way observable effects evolve when synchronous systems
are composed. To this end, we will define a small process algebra, inspired by
the category-theoretical work of Abramsky on concurrency \cite{AbramskyInteraction}.

\subsection{Data types}

In order to model multiple input-output ports, we will force a monoidal structure on the data 
processed by our synchronous systems. Let $Basic = \{ \inttype; \booltype; \singletonset; \ldots \}$ be
a set of basic datatypes. The set of datatypes is the monoid $\langle \datatypes, \times \rangle$ generated
by $Basic$ and closed by cartesian product.

\subsection{Composition operators}

Our composition operators are sequential composition and parallel composition. The transition relations
of the compound systems are defined
in a classic way, using a small-step semantics given by SOS inference rules (c.f. Fig. \ref{fig:composition}). 

\paragraph{Sequential composition.} Let $A, B, C$ be three sets.
Let $S_f = \langle A,$ $B,$ $Q_f,$ $E_f,$ $\outmorph_f,$ $q_{i,f} \rangle$ and
$S_g = \langle B,$ $C,$ $Q_g,$ $E_g,$ $\outmorph_g,$ $q_{i,g} \rangle$ be two systems.
The sequential composition $S_g \circ S_f$ proceeds by redirecting the output of $S_f$ to the input
of $S_g$.
The compound system is $S_g \circ S_f = \langle A, C, Q_f \times Q_g, E_{g \circ f}, \outmorph_{g \circ f}, (q_{i,f}, q_{i,g}) \rangle$,
where $E_{g \circ f}$ and $\outmorph_{g \circ f}$ are defined in Fig. \ref{fig:composition}.

\paragraph{Parallel composition.} Let $A, B, C, D$ be four sets.
Let $S_f = \langle A,$ $B,$ $Q_f,$ $E_f,$ $\outmorph_f,$ $q_{i,f} \rangle$ and
$S_g = \langle C,$ $D,$ $Q_g,$ $E_g,$ $\outmorph_g,$ $q_{i,g} \rangle$ be two systems.
The parallel composition proceeds by pairing the respective 
transitions of $S_f$ and $S_g$ in a synchronous way. The compound system is
$S_f \parallel S_g = \langle A \times C, B \times D, Q_f \times Q_g, E_{f \parallel g}, \outmorph_{f \parallel g}, (q_{i,f}, q_{i,g}) \rangle$
where $E_{f \parallel g}$ and $\outmorph_{f \parallel g}$ are also defined in Fig. \ref{fig:composition}.

\paragraph{Other operations.} An other important operation is the feedback. We omit it for space reasons,
but it must be noted that it exhibits the same behavior as sequential composition. 
The other operations necessary to make our definitions into an usable process algebra are
structural ones, like data duplication, erasure, etc. These important details are omitted from the following
study.

\begin{figure}

\begin{tabular}{m{7cm}m{7cm}}
  \hspace{1cm} \includegraphics{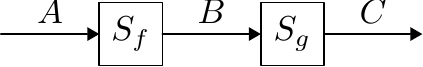} & \hspace{2cm} \includegraphics{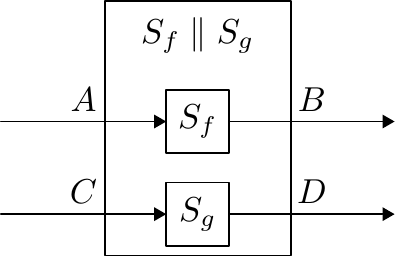} \\
  \begin{small}
  $
  \begin{array}{c}
    \inferrule [seq-next]
             {q_f \longedge{in_f} q'_f \\ q_g \longedge{\outmorph_f(q_f)} q'_g }
             { (q_f, q_g) \longedge{in_f} (q'_f, q'_g) } \\ \\

    \inferrule [seq-output]
             { }
             { \outmorph_{g \circ f}(q_f, q_g) = \outmorph_g(q_g) }
  \end{array}
  $ \end{small} & \begin{small}
  $
  \begin{array}{c}
    \inferrule [par-next]
             {q_f \longedge{in_f} q'_f \\ q_g \longedge{in_g} q'_g }
             { q_f \parallel q_g \longedge{\langle in_f, in_g \rangle} q'_f \parallel q'_g} \\ \\

    \inferrule [par-output]
             { }
             { \outmorph_{f \parallel g}(q_f \parallel q_g) = \langle \outmorph_f(q_f), \outmorph_g(q_g) \rangle}
  \end{array}
  $ \end{small} \\
  \hspace{1.5cm} (a)~Sequential composition & \hspace{1.5cm} (b)~Parallel composition
\end{tabular}
\caption{Definition of the composition operations}
\label{fig:composition}
\end{figure}

\subsection{Observable effects w.r.t. sequential composition}

In this section, we study the behavior of the observable effects of systems when they are
composed. We restrict our attention to sequential composition since it is easy to show that
parallel composition doesn't alter the behavior of the sub-components.
We show that under sequential composition, whenever reactivity still holds, 
the observable effects can vary arbitrarily. Our examples will be given on Moore machines
whose state space is not quotiented by bisimulation equivalence.

\vspace{3mm}
\hspace{-9mm}
\begin{tabular}{m{7cm}m{8cm}}

The proof that reactivity can be lost follows the same argument that shows that the
composition of two non-constant total functions can be constant. The figure to the
right shows the composition of two reactive Moore machines $S_f$ and $S_g$ whose composition is
not reactive. In this example, this stems from the fact that the observable effect $(0, 1)$
of the input received in state $p_0$ of the machine $S_f$ is not ``taken into account'' by the machine $S_g$, i.e.
$(0, 1)$ is not a separating pair of state $q_1$.

& 

\begin{tabular}{c}
  \begin{tabular}{m{2.3cm}m{4.3cm}}
    \includegraphics[width=2.3cm]{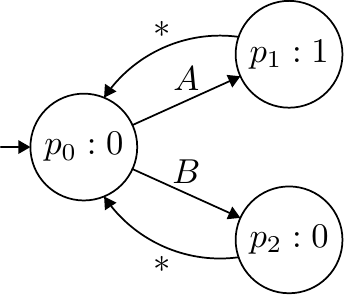} &
    \includegraphics[width=4.3cm]{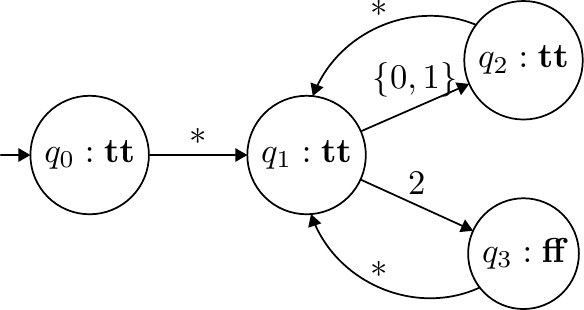} \\
    \hspace{6mm} Machine $S_f$ &  \hspace{20mm} Machine $S_g$
  \end{tabular} \\
  \includegraphics[width=1.0cm]{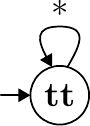} \\
  Machine $S_g \circ S_f$
\end{tabular}

\end{tabular}

\begin{figure}
\centering
\begin{tabular}{cc}
  \includegraphics[width=5cm]{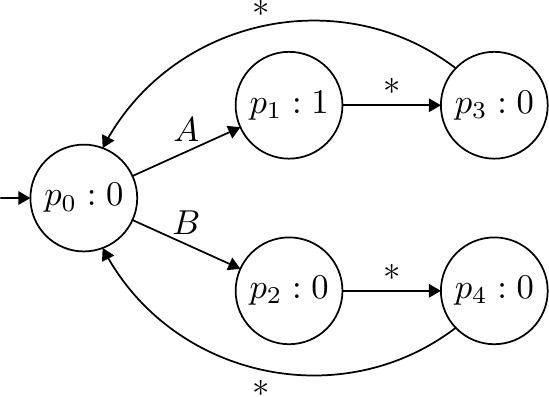} \hspace{0.5cm} &
  \hspace{0.5cm} \includegraphics[width=5cm]{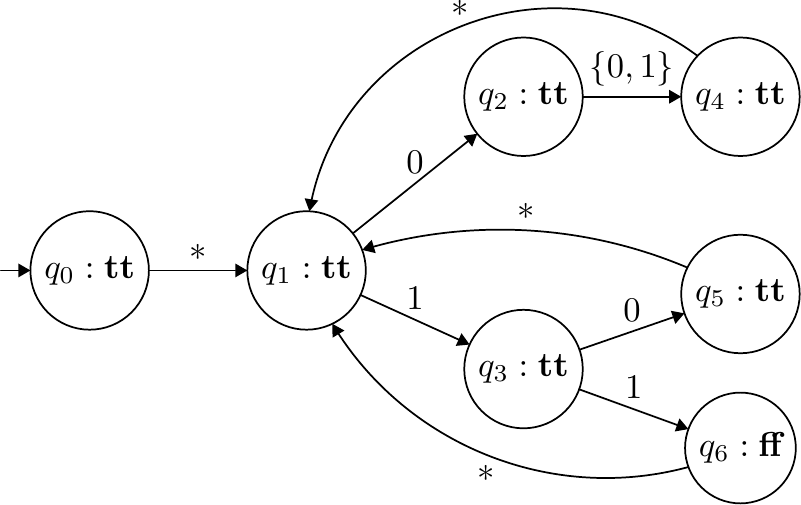} \\
  Machine $S_f$ & Machine $S_g$
\end{tabular}
\caption{Example of disappearing separator}
\label{fig:disappearingseparators}
\end{figure}

The fact that an observable effect of $S_f$ is a separating pair of $S_g$ is not enough to
guarantee an observable effect on output. Sequential composition restricts the input
language of the system in receiving position (here, $S_g$). This means that separators
can appear and disappear arbitrarily. The two Moore machines in Fig.
\ref{fig:disappearingseparators} are modifications of the earlier ones.
The states $p_0$ and $q_1$ are still reactive, but when composed
the output symbols on $p_3$ and $p_4$ restricts the set of inputs
of the states $q_2$ and $q_3$ to the word $0$. Thus, $q_2$ and $q_3$
are no more separable and the result is the constant machine
shown earlier.

The conclusion of this study confirms the intuition: there is no general way of
guaranteeing functional dependencies. These results extend to reaction time, 
which is not conserved: the receiving machine may ignore the first observable
effect and take into account ulterior ones.

In verification terms, this means that in order to verify that the composition 
of two systems is reactive, a full search of the state space for separators must be undertaken. 
In the next section, an approximate but compositional method to simplify 
this process is proposed.

\section{Under-approximating observable effects}

This section proposes a partial solution to some problems encountered earlier, namely:
\begin{enumerate}
\item the fact that non-deterministic separators do not guarantee an observable effect,
\item the non-compositionality of reactivity and observable effects.
\end{enumerate}

We proceed by reducing our focus to the cases where reactivity, which is a branching-time property of states,
can be reduced to a linear-time one. We show how to compute the separators and separating pairs which are preserved when
``merging'' all branches of the computation tree.

Let us assume the existence of two sets of data $In$ and $Out$.
Let $q$ be a state such that $\reactive(q)$ holds, and let
$(a_1, a_2) \in \separatingpairs(s)$ be a separating pair of inputs.
In Sec. \ref{sec:reactiontime}, we observed that in order to ensure the occurrence of 
an observable effect and the existence of a reaction time, $q$ must be such that 
all inputs are deterministic separators for all $q_1$ and $q_2$ s.t. $q \longedge{a_1} q_1$ and
$q \longedge{a_2} q_2$. If this condition is met, we can compute \emph{deterministic observable
effects}, i.e. effects which exists for \emph{all} separators. Similarly, we can define
\emph{deterministic separating pairs}.

First, we define some operations in order to merge sequences of observable effects.
We define the operation $\merge : \left( \singletonset + Out \times Out \right) \times \left( \singletonset + Out \times Out \right) \rightarrow \left( \singletonset + Out \times Out \right)$ as:
$$
\begin{array}{ccccll}
x & \merge & x & = & x \\
x & \merge & y &  = & \star &\text{if $x \neq y$.}
\end{array}
$$
The extension of this operation to sequences of symbols on $(\singletonset + Out \times Out)$ is defined
straightforwardly. If $d_1, d_2 \in (\singletonset + Out \times Out)^\omega$ are two infinite sequences,
their merging is also noted $d_1 \merge d_2$.

\begin{definition}[Deterministic observable effects, observational order]
  Let $q$ be a state s.t. $\reactive(q)$ holds. The sequence of deterministic observable effects of $q$
  is noted $\detobs(q)$ and is defined as follows:
  $$
  \detobs(q) = \begin{array}{c} \mbox{\Huge $\oplus$} \\ {\scriptstyle (a_1, a_2) \in \separatingpairs(q)} \end{array}
  \{ \differences_{q_1, q_2}(w)~|~q \longedge{a_1} q_1, q \longedge{a_2} q_2, w \in In^\omega \}.
  $$
  It is possible to define a relation $\prec~ \subseteq (\singletonset + Out \times Out)^\omega \times (\singletonset + Out \times Out)^\omega$, where:
  $$
  w_1 \prec w_2 \leftrightarrow \exists! i, \left( \forall j \neq i, w_1[j] = w_2[j] \right) \wedge (w_1[i] = \star \wedge w_2[i] \neq \star).
  $$
  The reflexive-transitive closure of $\prec$ is the \textbf{observational order} and is noted $\DOEorder$. The set $\obsorder(q)$ of infinite
  strings partially ordered by $\DOEorder$ which has $\detobs(q)$ as greatest element and $\star^\omega$ as least element is called
  by extension the observational order on $q$.
\end{definition}

We must also define linear time-proof separating pairs, called \emph{strongly separating pairs}. Let's consider the systems 
in Fig. \ref{fig:union}, in which only the output data is displayed and state names are omitted. The systems 1 and 2 are symmetrical 
and have both $(A, B)$ as a separating pair for their initial state. However, $(A, B)$ is \emph{not} a separating pair for the \emph{union}
of the two systems. We must define a notion of separating pair for two systems which resists their union.

\begin{figure}[!h]
  \centering
  \subfigure[System 1]{
    \includegraphics[width=4cm]{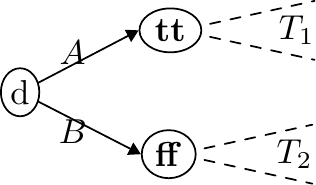}
    \label{fig:union:little3}
  }
  \subfigure[System 2]{
    \includegraphics[width=4cm]{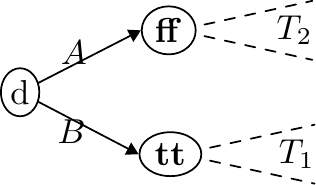}
    \label{fig:union:little4}
  }
  \subfigure[Union of systems 1 and 2]{
    \includegraphics[width=4cm]{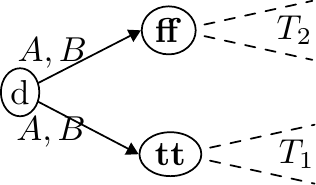}
    \label{fig:union:little5}
  }
  \caption{Union of transition systems}
  \label{fig:union}
\end{figure}

\begin{definition}[Strongly separating pairs]
  Let $q_1, q_2$ be two states s.t. $\reactive(q_{1,2})$ holds. The pair $(a_1, a_2) \in \separatingpairs(q_{1,2})$
  is a strongly separating pair of the union of $q_1$ and $q_2$ if and only if:
  $$
  \left( \exists q_1 \edge{a_1} q'_1, \forall q_2 \edge{a_2} q'_2, q'_1 \nsim q'_2 \right) \vee \left( \exists q_1 \edge{a_2} q'_1, \forall q_2 \edge{a_1} q'_2, q'_1 \nsim q'_2 \right).
  $$
  The set of strongly separating pairs of $q_1$ and $q_2$ is noted $\strongseppairs(q_1, q_2)$. Note
  that $\separatingpairs(q) = \strongseppairs(q, q)$. The \emph{sequence} of strong separating pairs of $q_1$ and $q_2$
  is noted $\sspsequence(q_1, q_2)$ and is defined as follows:
  $$
  \sspsequence(q_1, q_2) = \strongseppairs(q_1, q_2) . \left( \overline{\bigcap} \{  \sspsequence(q'_1, q'_2)~|~(a_1, a_2) \in \separatingpairs(q), q_1 \edge{a_1} q'_1, q_2 \edge{a_2} q'_2\} \right),
  $$
  where $\overline{\bigcap}$ is the extension of set intersection to sequences of sets.
  Now, let $q$ be s.t. $\reactive(q)$ holds. The sequence of strongly separating pairs of $q$
  is:
  $$
  \sspsequence(q) = \separatingpairs(q) . \overline{\bigcap} \{ \sspsequence (q_1, q_2)~|~(a_1, a_2) \in \separatingpairs(q), q \edge{a_1} q_1, q \edge{a_2} q_2 \}.
  $$
\end{definition}

Deterministic observable effects are in fact an \emph{abstraction} of the original
system. The concretization operation is the function associating to a sequence of 
deterministic observable effects the set of all systems which have at least these deterministic 
observable effects (w.r.t. $\DOEorder$). Using this abstraction, checking the compositionality of sequential composition is straightforward.

\begin{lem}[Sequential composition of deterministic observable effects ensures reactivity]
  Let $S_f = \langle A,$ $B,$ $Q_f,$ $E_f,$ $\outmorph_f,$ $q_{i,f} \rangle$ and
  $S_g = \langle B,$ $C,$ $Q_g,$ $E_g,$ $\outmorph_g,$ $q_{i,g} \rangle$ be two systems.
  Let $q_f \in Q_f$ and $q_g \in Q_g$ be two states such that $(q_f, q_g)$ is in the state space of the sequential
  composition $S_g \circ S_f$. We have $\reactive(q_f, q_g)$ if:
  $$
  \exists d \in \obsorder(q_f), \exists i, d[i] \in \sspsequence(q_g)[i+1].
  $$
\end{lem}

\begin{proof}
  Let $i \in \natset$ be such that $d[i] \in \sspsequence(q_g)[i+1]$. Having $d[i] = (x_1, x_2)$
  implies that $\reactive(q_f)$ holds. Hence, there exists $a_1 \neq a_2$ s.t. 
  $\exists q_f \edge{a_1} q^1_f, \forall q_f \edge{a_2} q^2_f, q^1_f \nsim q^2_f$. By hypothesis, we know that all input words are
  separators of all such $(q^1_f, q^2_f)$. By definition, $d[i]$ is an observable effect of all these separators.
  Hence, for all input words $w$ there will exist two runs $q^1_f \transedge{\prefix(w,i)} r^1_f$ 
  and $q^2_f \transedge{\prefix(w,i)} r^2_f$ s.t. $\outmorph(r^1_f) = x_1$ and $\outmorph(r^2_f) = x_2$. 
  By definition of the sequential composition, this induces the runs $(q^1_f, q^1_g) \transedge{\prefix(w,i)} (r^1_f, r^1_g)$ and
  $(q^2_f, q^1_g) \transedge{\prefix(w,i)} (r^2_f, r^2_g)$ (with $q_g \longedge{\outmorph(f)} q^1_g$). 
  Since $(x_1, x_2) \in \strongseppairs(r^{1,2}_g)$ (by definition of $\sspsequence$),  $(q^1_f, q^1_g) \nsim (q^2_f, q^1_g)$.
\end{proof}

The compositionality of this approach stems from the fact that for any
systems $S_f, S_g$ and respective states $q_f$ and $q_g$, an element of $\obsorder(q_f, q_g)$ can be computed using 
other elements from $\obsorder$ and $\strongseppairs$. 

\begin{definition}[Compositionality of deterministic observable effects]
  Let $q_f \in Q_f, q_g  \in Q_g$ be two states.
  Let $doe_f \in \obsorder(q_f)$ and $dsp_g \in \strongseppairs(q_g)$.
  If there exists a $t$ s.t. $doe_f[t] \in dsp_g[t+1]$ then there exists an element $doe^t_{g \circ f} \in \obsorder(q_f, q_g)$
  s.t.:
  $$
  doe^t_{g \circ f} = \star . \star^t . \mbox{\Huge $\oplus$} \{ doe_{g'}~|~ (q_f, q_g) \rightarrow^{t+1} (q'_f, q'_g), doe_{g'} \in \obsorder(q'_g) \}.
  $$
  This element is computed by merging the deterministic observable effects of states of $q_g$ reachable in $t+1$
  transitions. The initial $\star$ stems from the delay induced by communication in the synchronous model.
\end{definition}

A similar property holds for separating pairs. We only describe informally how to proceed,
since the general idea is similar to the case of deterministic observable effects.
A strongly separating pair of $q_f$ exists in $S_g \circ S_f$ if the
states $q^1_f, q^2_f$ reachable by this pair have a common observable effect (computable
using $\merge$) which corresponds to a strongly separating pair of $q_g$.

\section{Example}

As explained in the introduction, our work focuses on a real-time system called OASIS,
which provides a real-time kernel and a multi-agent synchronous-like language called
PsyC (an extension of C with synchronous primitives). We have given a formal semantics
of a simplification of PsyC called Psy-ALGOL.
We will use this semantics to highlight a common use case of our framework.

\subsection{Syntax and semantics of a simple synchronous language.}

In order to give the semantics of a program, we have to define how
its LTS is generated. We briefly survey a subset of the syntax of the language. 
The connection between the semantics and the resulting LTS should be 
straightforward, we will thus omit the derivations.

\subsubsection{Syntax.} The syntax definition is given in an inductive way using inference rules on
judgments of the shape  $\Gamma \vdash M : \sigma$, meaning ``in the context $\Gamma$,
the program $M$ has type $\sigma$''. A context is a list of the shape
$x_1 : \sigma_1, \ldots, x_n : \sigma_k$. It associates variables $x_i$ to their types $\sigma_i$.
For the sake of simplicity, we assume that all variables are declared beforehand and initialized
to their default value. We ignore procedures and we keep the other syntactical forms as simple as possible.
The types $\sigma$ and default values are defined as follows:
$$
\begin{array}{ll}
  \begin{array}{lll}
    \tau   & ::= & \inttype  ~|~ \booltype \\
    \sigma & ::= & \commtype ~|~ \vartype{\tau} ~|~ \exptype{\tau}
  \end{array} &
  \begin{array}{lll}
    default_\inttype & = & 0 \\
    default_\booltype & = & \Afalse \\
  \end{array}
\end{array}
$$

\begin{figure}
\begin{scriptsize}
$$
\begin{array}{lccccr}
  \inferrule*
  { }
  { \Askip : \commtype } &
  
  \inferrule*
  { }
  { \Gamma, x : \sigma \vdash x : \sigma } &

  \inferrule*
  { b \in \{ \Atrue, \Afalse \} }
  { \Gamma \vdash b : \boolexptype } &

  \inferrule*
  { n \in \natset }
  { \Gamma \vdash n : \intexptype } &

  \inferrule*
  { \Gamma \vdash V : \vartype{\tau} }
  { \Gamma \vdash !V : \exptype{\tau} } &

  \inferrule*
  { \Gamma \vdash V : \vartype{\tau} \\ \Gamma \vdash E : \exptype{\tau} }
  { \Gamma \vdash V := E : \commtype }

\end{array}
$$
$$
\begin{array}{lcr}

  \inferrule*
  { \Gamma \vdash A_0 : \exptype{\booltype} \\ \Gamma \vdash A_0 : \sigma \\ \Gamma \vdash A_1 : \sigma }
  { \Gamma \vdash \Acond{A_0}{A_1}{A_2} : \commtype } &

  \inferrule*
  { \Gamma \vdash A_0 : \commtype \\ \Gamma \vdash A_1 : \sigma }
  { \Gamma \vdash A_0; A_1 : \commtype } &

  \inferrule*
  { \Gamma \vdash A_0 : \exptype{\booltype} \\ \Gamma \vdash A_1 : \commtype }
  { \Gamma \vdash \Awhile{A_0}{A_1} : \commtype } \\

\end{array}
$$
$$
\begin{array}{cc}

  \inferrule*
  { \Gamma \vdash A_0 : \pi_0~Out \\ \ldots \\ \Gamma \vdash A_{|Out|-1} : \pi_{|Out|-1}~Out }
  { \Gamma \vdash \Atick{A_0 \ldots A_{|Out|-1}} : \commtype } &

  \inferrule*
  { }
  { \Gamma \vdash \Aget{i} : \pi_i~In \\ i \in \fin{|In|} }

\end{array}
$$
\end{scriptsize}
\caption{Definition of well-typed programs.}
\label{syntax}
\end{figure}
Assuming that the programs have input and output types $\langle In, Out \rangle$, the set $Prog$ of
correctly typed programs is defined in Fig. \ref{syntax}. We omit the arithmetical operators.

\subsubsection{Semantics.} \label{semantics} We will define the operational semantics for our 
language as a small-step relation. An operational semantics is usually
a kind of relation associating a program in its initial configuration to its
final outcome (be it a final value or divergence).

Our aim is slightly different: we want to view the evaluation of a program as a 
synchronous system.
This means that instead of producing a final value or diverging, we want to quantify
over all possible inputs at each logical step, and produce a LTS. In order to simplify matters,
our LTS will be given in unfolded form, as an infinitely deep tree.
Each step of the evaluation will grow this tree downward, and the limit of this process
will be the semantics of the program. Let's proceed to some definitions.

\begin{definition} The set of configurations is $Config \define Store(V) \times In \times Prog$,
  where $V$ is the set of variables of the program and $Store(V)$ is a mapping from variables to
  constants.
\end{definition}

\begin{definition} The sets $Trees$ of finite (resp. infinite) partial evaluation trees are generated by the inductive 
(resp. co-inductive) interpretation of the following rules.
$$
\begin{array}{cc}
  \inferrule[Leaf]
  { conf \in Config  }
  { conf \in Trees } &
  \inferrule[Node]
  { out \in Out \\ \forall in, tr_{in} \in Trees }
  { (out, \{ (in, tr_{in}) ~|~ in \in In \}) }
\end{array}
$$
Let $T$ be a partial evaluation tree with leaves $(st_i, in_i, prog_i)$.
Given a syntactical mapping $\mapfunc : Prog \rightarrow Prog$, we note $\mapfunc \downarrow T$ the extension
of $\mapfunc$ to the leaves of $T$ such that $(\mapfunc \downarrow T)$ has leaves $(st_i, in_i, \mapfunc(prog_i))$.

\end{definition}

The one-step reduction relation $\onestep ~\subseteq FiniteTrees \times FiniteTrees$
is defined in terms of a relation $\preonestep ~\subseteq Config \times FiniteTrees$
defined by the rules in Fig.~\ref{pre_onestep}.

\begin{figure}[h]
\begin{scriptsize}
$$
\begin{array}{c}
\begin{array}{lll}
  \inferrule[Const]
  { }
  { (st, in, c) \preonestep (st, in, c) \\ c \in \values } &

  \inferrule[Var]
  { }
  { (st, in, x) \preonestep (st, in, x) \\ x \in dom(st) } &

  \inferrule[Seq-context]
  { (st, in, A_0) \preonestep T }
  { (st, in, A_0; A_1) \preonestep (\Abs{A'_0} A'_0; A_1) \downarrow T } \\ \\

  \inferrule[Seq-skip]
  {  }
  { (st, in, \Askip; A_1) \preonestep (st, in, A_1) } &

  \inferrule[Deref-context]
  { (st, in, A) \preonestep T }
  { (st, in, !A) \preonestep (\Abs{A'} !A') \downarrow T } &

  \inferrule[Deref-var]
  { }
  { (st, in, !x) \preonestep (st, in, st(x)) } \\ \\

  \inferrule[Assign-context]
  { (st, in, A) \preonestep T }
  { (st, in, x := A) \preonestep (\Abs{A'} x := A') \downarrow T } &

  \inferrule[Assign-const]
  { }
  { (st, in, x := v) \preonestep (st | x \mapsto v, in, \Askip) } &

  \inferrule[Get]
  { }
  { (st, in, \Aget{i}) \preonestep (st, in, \pi_i~in) }

\end{array} \\ \\

\begin{array}{c}
  \inferrule[If-cond-context]
  { (st, in, A_0) \preonestep T }
  { (st, in, \Acond{A_0}{A_1}{A_2}) \preonestep (\Abs{A'_0} \Acond{A'_0}{A_1}{A_2}) \downarrow T }
\end{array} \\ \\

\begin{array}{ll}
  \inferrule[If-true]
  { }
  { (st, in, \Acond{\Atrue}{A_1}{A_2}) \preonestep (st, in, A_1) } &

  \inferrule[If-false]
  { }
  { (st, in, \Acond{\Afalse}{A_1}{A_2}) \preonestep (st, in, A_2) } \\ \\
\end{array} \\

\begin{array}{c}
  \inferrule[While-unfold]
  { }
  { (st, in, \Awhile{A_0}{A_1}) \preonestep (st, in, \Acond{A_0}{A_1; \Awhile{A_0}{A_1}}{\Askip}) } \\ \\
\end{array} \\

\inferrule[Tick-context-i]
{ (st, in, A_i) \preonestep T }
{ (st, in, \Atick{c_0~\ldots~c_{i-1}, A_i~\ldots A_{|Out|-1}}) \preonestep (\Abs{A'_i} \Atick{c_0~\ldots~c_{i-1}, A'_i~\ldots A_{|Out|-1}}) \downarrow T } \\ \\

\inferrule[Tick-constant]
{ }
{ (st, in, \Atick{c_0~\ldots c_{|Out|-1}}) \preonestep ((c_0~\ldots c_{|Out|-1}), \{ (input, (st, input, \Askip)) ~|~ input \in In \}) }

\end{array}
$$
\end{scriptsize}
\caption{One-step reduction relation}
\label{pre_onestep}
\end{figure}

We define $\onestep$ as the application of $\preonestep$ to the leaves of a tree.
From there, we can define the standard reflexive-transitive closure of $\onestep$
and its co-inductive counterpart as in \cite{LeroyBigStep}. 

\subsection{Example of synchronous programs.}

\begin{figure}
  \centering
  \subfigure[Programs 1 and 2]{
    \begin{tabular}{ll}
$\begin{array}{l}
    \textsc{Program 1} \\
    x := \Afalse; \\
    \textbf{while}~\Atrue~\textbf{do} \\
    \hspace{5mm} \redt{\Atick{!x}}; \redt{\Leftarrow}\\
    \hspace{5mm} x~:=~\Agetn; \\
    \hspace{5mm} \textbf{while}~\Agetn~\textbf{do} \\
    \hspace{10mm}   \Atick{\Afalse} \\
    \hspace{5mm}\textbf{done}; \\
    \textbf{done}
  \end{array}$ &
$\begin{array}{l}
    \textsc{Program 2} \\
    x := \Afalse; \\
    y := N; \\
    \textbf{while}~\Atrue~\textbf{do} \\
    \hspace{5mm} \redt{\Atick{!x}}; \redt{\Leftarrow} \\
    \hspace{5mm} x~:=~\Agetn; \\
    \hspace{5mm} y~:=~N; \\
    \hspace{5mm} \textbf{while}~\Agetn~\wedge~!y~\neq 0~\textbf{do} \\
    \hspace{10mm}   y := !y - 1; \\
    \hspace{10mm}   \Atick{\Afalse} \\
    \hspace{5mm}\textbf{done}; \\
    \textbf{done}
  \end{array}$
    \end{tabular}
  }
  \subfigure[LTS of program 1]{
    \includegraphics[width=13.0cm]{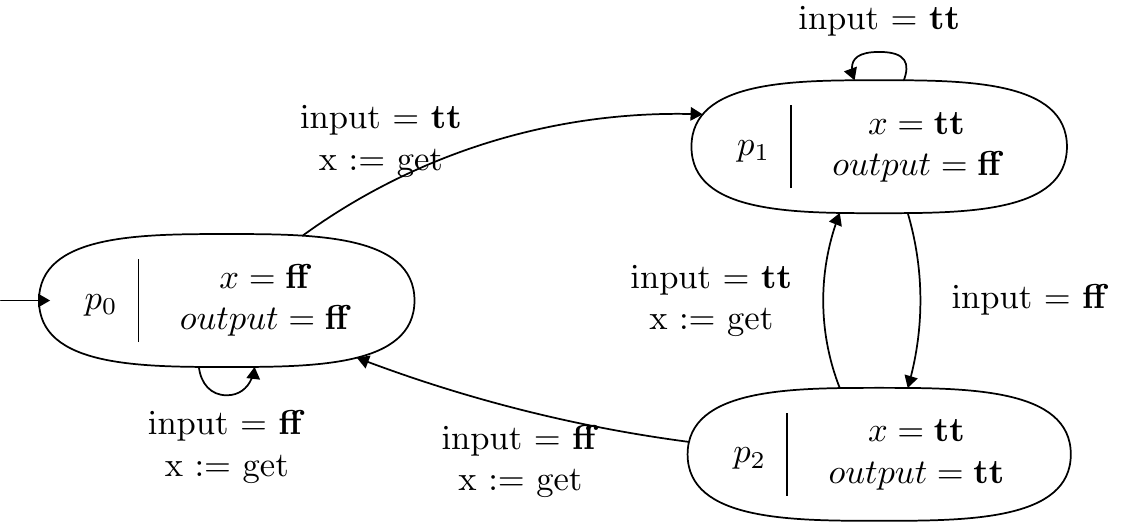}
    \label{fig:LTS:prog1}
  }
  \subfigure[LTS of program 2]{
    \includegraphics[width=15.0cm]{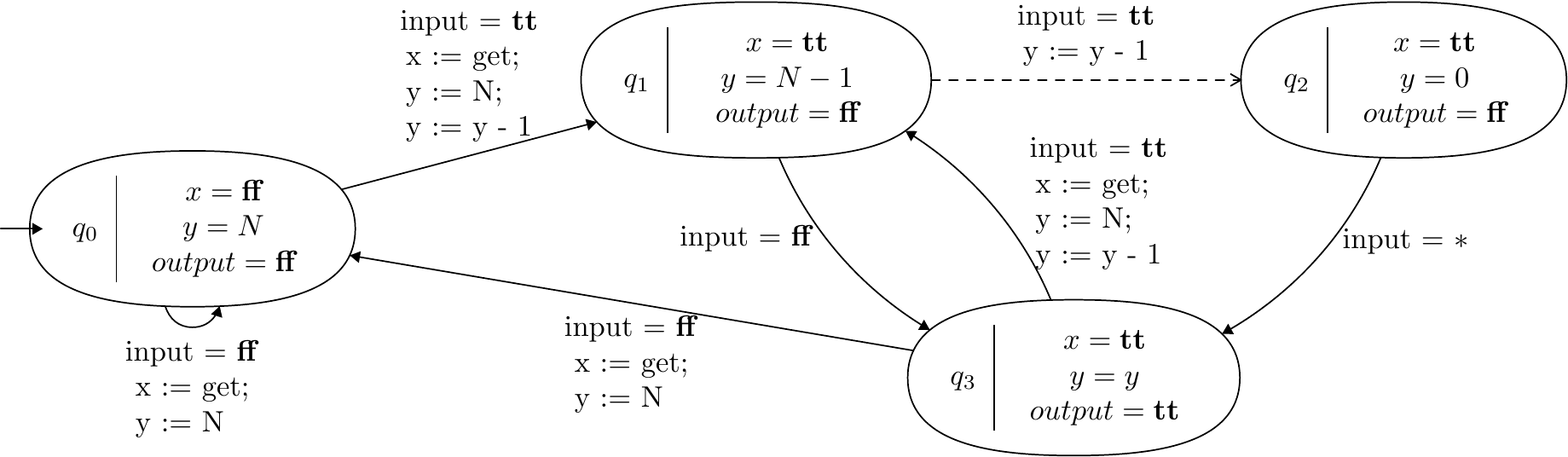}
    \label{fig:LTS:prog2}
  }
  \caption{LTS of programs 1 and 2}
  \label{fig:LTS}
\end{figure}

We will study the behavior of two programs, whose texts and associated (deterministic) LTS are displayed in Fig. \ref{fig:LTS}.
The first program captures an input data at the beginning of the outer loop and releases it on output when the inner loop
finishes executing itself. The second program proceeds similarly, except that the inner loop termination is ensured by
the usage of a decreasing counter $y$ initialized to a constant $N$. In the LTS of program 2, this corresponds to the
dashed transition between $q_1$ and $q_2$, which should be understood as $N - 2$ omitted states with decreasing values
of $y$.

We want to check whether the data inputted at the $\Atick{!x};$ lines highlighted in both programs
yield a finite reaction time. These instructions corresponds to states $p_0, p_2$ in LTS 1, and $q_0, q_3$ in LTS 2.
Thus, in order for these instructions to be ``reactive'', the corresponding states must have a finite reaction time.
In order to study this, we will compute their deterministic observable effects.

\paragraph{Program 1.} The set of separating pairs of $p_0$ and $p_2$ is $\separatingpairs(p_{0,2}) = \{ (\Atrue, \Afalse) \}$.
This fact is proved by the transitions $p_{0} \longedge{\Atrue} p_1$, $p_{0} \longedge{\Afalse} p_0$,
and $p_2 \longedge{\Atrue} p_1$, $p_2 \longedge{\Afalse} p_0$ where $p_0 \nsim p_1$.  The only separator of $(p_0, p_1)$ is
the one-symbol word $w = \Afalse$. This fact is proved by the transitions $p_0 \longedge{\Afalse} p_0$ and $p_1 \longedge{\Afalse} p_2$,
where $\outmorph(p_0) \neq \outmorph(p_2)$. The separator $w$ is \emph{deterministic}, since the underlying automaton is itself
deterministic. This separator induces a pair of output words $(o_1, o_2) = (\Afalse . \Afalse, \Afalse . \Atrue)$ and thus an observable
effect $\differences_{p_0,p_1}(w, 1) = (\Afalse, \Atrue)$. However, this observable effect is \textbf{not} deterministic, since there
exist an infinite input word $\Atrue^\omega$ which generates no observable effect. Thus, the sequence of observable
effects of $p_0$ and $p_2$ is $\star^\omega$ and the reaction time for the highlighted line instruction does not exist (or, equivalently, is infinite).

\paragraph{Program 2.} The set of separating pairs of $q_0$ and $q_3$ is still 
$\separatingpairs(q_{0,3}) = \{ (\Atrue, \Afalse) \}$. The corresponding transitions are 
$q_0 \longedge{\Atrue} q_1$, $q_0 \longedge{\Afalse} q_0$ and 
$q_3 \longedge{\Atrue} q_1$, $q_3 \longedge{\Afalse} q_0$ with $q_0 \nsim q_1$. 
The (deterministic) separators of $(q_0, q_1)$ are 
$\nondetseparators(q_0, q_1) = \detseparators(q_0, q_1) = \{ \textbf{ff}; ~ \textbf{tt} . \textbf{ff}; ~ \textbf{tt} .  \textbf{tt} . \textbf{ff}; ~ \ldots; ~ \textbf{tt}^{N-2} . \textbf{ff};~
\textbf{tt}^{N-1} . \booltype \}$. The table below lists the observable differences associated to each separator.

$$
\begin{array}{lll}
  \Afalse & \mapsto & \star . (\Afalse, \Atrue) . \star^\omega \\
  \textbf{tt} . \textbf{ff} & \mapsto & \star . \star . (\Afalse, \Atrue) . \star^\omega \\
  & \ldots & \\
  \textbf{tt}^{N-2} . \textbf{ff} & \mapsto & \star^{N-1} . (\Afalse, \Atrue) . \star^\omega \\
  \textbf{tt}^{N-1} . \booltype   & \mapsto & \star^N . (\Afalse, \Atrue) . \star^\omega \\
\end{array}
$$

When merging these observable differences, we obtain $\detobs(q_{0,3}) = \star^\omega$. This means that even though
the program 2 is reactive with a finite reaction time, it is still non-compositional within our simple
framework. This is due to the fact that the observable effects occurrence time are non-uniform w.r.t. inputs, i.e. 
non-constant.

\section{Conclusions and future works}

We have formalized in this paper the notions of functional dependency and reaction time 
for some synchronous systems. These notion are adapted to the formal investigation of reaction time constraints 
for the aforementioned synchronous systems. Functional dependencies were shown to be brittle and not 
suited to composition and verification. To answer this problem, we proposed an approximated method 
which gains compositionality by restricting its scope to deterministic separators.

Our work opens some other research directions: a broader investigation of the notion of reaction 
time in a more general setting \cite{Haghverdi2005} could prove fruitful and lead to simpler, 
more abstract and general definitions. Our composition operators are quite restricted, as shown
in the example, but making it more flexible should be possible by making deterministic effects
a function of some arbitrary decidable specification. It seems also possible to apply our ideas to the
refinement-based development of systems. 

We provide a formal framework allowing to reason on functional dependencies and reaction time
which is amenable to automated verification. This effort should help the software designer and 
programmer to deliver reliable, predictable and efficient systems.

\bibliographystyle{eptcs}
\bibliography{ice}

\end{document}